\documentclass[12pt]{article}
\usepackage[utf8]{inputenc}
\usepackage{amssymb}
\usepackage{amsfonts}
\usepackage{amsmath}
\usepackage{amsthm}
\usepackage{mathrsfs} 
\usepackage{graphicx}
\usepackage{mathtools}
\numberwithin{equation}{section}
\usepackage{xcolor}

\usepackage{yfonts}
\usepackage{amssymb}
\usepackage{amsmath}
\theoremstyle{theorem}
\newtheorem{theorem}{Theorem}[section]
\newtheorem{lemma}[theorem]{Lemma}
\newtheorem{proposition}[theorem]{Proposition}

\newtheorem{definition}[theorem]{Definition}
\newtheorem{remark}[theorem]{Remark}

\newtheorem{example}[theorem]{Example}

\def\sp{\hskip -5pt}
\def\spa{\hskip -3pt}

\def\bearray{\begin{eqnarray}}
\def\earray{\end{eqnarray}}
\def\beq{\begin{equation}}
\def\eeq{\end{equation}}

\def\b0{{\bf 0}}

\def\mpasto{\mapsto}

\newsymbol\bt 1202             

\def\bC{{\mathbb C}}           

\def\bN{{\mathbb N}}

\def\bR{{\mathbb R}}

\newsymbol\rest 1316         

\begin{document} 

\par
\bigskip
\large
\noindent
{\bf Bulk-boundary asymptotic equivalence of two strict deformation quantizations}
\bigskip
\par
\rm
\normalsize

\noindent  {\bf Valter Moretti$^{a}$$^1$}, {\bf Christiaan J.F.  van de Ven$^{a}$$^2$}\\
\par

\noindent 
 $^a$ Department of  Mathematics, University of Trento, and INFN-TIFPA \\
 via Sommarive 14, I-38123  Povo (Trento), Italy\\

\smallskip
\noindent $^1$ Email: valter.moretti@unitn.it,  $^2$ Marie Sk\l odowska-Curie Fellow of the Istituto Nazionale di Alta Matematica. Email:
christiaan.vandeven@unitn.it\\
 \normalsize

\par

\rm\normalsize

\spa\spa\sp\sp September 21 2020

\rm\normalsize


\par
\bigskip
\noindent
\small
{\bf Abstract}.  
The existence of a strict deformation quantization of  $X_k=S(M_k({\mathbb{C}}))$, the state space of the $k\times k$ matrices $M_k({\mathbb{C}})$ which  is canonically a compact Poisson manifold (with stratified boundary)  has recently been proven by both authors and K. Landsman \cite{LMV}.
In fact, since increasing tensor powers of the $k\times k$ matrices $M_k({\mathbb{C}})$ are known to give rise to a continuous bundle of $C^*$-algebras over 
$I=\{0\}\cup 1/\mathbb{N}\subset[0,1]$ with fibers $A_{1/N}=M_k({\mathbb{C}})^{\otimes N}$ and $A_0=C(X_k)$,  we were able to define a strict deformation quantization of $X_k$ \`{a} la Rieffel, specified by quantization maps $Q_{1/N}: \tilde{A}_0\rightarrow A_{1/N}$, with $\tilde{A}_0$ a  dense Poisson subalgebra of $A_0$. A similar result is known for the symplectic manifold $S^2\subset\mathbb{R}^3$, for which in this case the fibers $A'_{1/N}=M_{N+1}(\mathbb{C})\cong B(\text{Sym}^N(\mathbb{C}^2))$ and $A_0'=C(S^2)$ form a continuous bundle of $C^*$-algebras over the same base space $I$, and where quantization is specified by (a priori different) quantization maps $Q_{1/N}': \tilde{A}_0' \rightarrow A_{1/N}'$. In this paper we focus on the particular case  $X_2\cong B^3$ (i.e the unit three-ball in $\mathbb{R}^3$) and show that for any function $f\in \tilde{A}_0$ one has $\lim_{N\to\infty}||(Q_{1/N}(f))|_{\text{Sym}^N(\mathbb{C}^2)}-Q_{1/N}'(f|_{_{S^2}})||_N=0$, were $\text{Sym}^N(\mathbb{C}^2)$ denotes the symmetric subspace of  $(\mathbb{C}^2)^{N \otimes}$. Finally, we give an application regarding the (quantum) Curie-Weiss model.
\normalsize

\tableofcontents

\section{Introduction}
An important field of research within mathemetical physics concerns the relation between classical theories viewed as limits of quantum theories.  For example, classical mechanics of a particle on the phase space $\mathbb{R}^{2n}$ versus quantum mechanics on the Hilbert space $L^2(\mathbb{R}^n)$, or classical thermodynamics of  a spin system versus statistical mechanics of a quantum spin system on a finite lattice \cite{Lan17}.  In these examples the relation between both (different) theories can be described by a {\em continuous bundle} of  algebras of observables  equipped with certain quantization maps. A modern way establishing a link  between both theories is based on the concept of {\em strict deformation quantization}, i.e., the mathematical formalism that describes the transition from a classical theory to a quantum theory \cite{Rie89,Rie94,Lan98} in terms of  deformations of (commutative) Poisson algebras (representing the classical theory) into non-commutative $C^*$  algebras characterizing the quantum theory.  
\subsection{Strict deformation quantization maps}
Let us focus to the first known  example starting from the familiar classical phase space  $\mathbb{R}^{2n}$. For convenience, we only consider the Poisson algebra of smooth compactly supported functions $f\in C_c^{\infty}(\mathbb{R}^{2n})$ where the Poisson structure is the one associated to the natural symplectic form $\sum_{j=1}^n dp_j \wedge dq^j$. In order to relate $C_c^{\infty}(\mathbb{R}^{2n})$  to a quantum theory described on some Hilbert space, one needs to deform $C_c^{\infty}(\mathbb{R}^{2n})$ into non-commutatative $C^*$-algebras exploiting  a family of {\em quantization maps}. Berezin proposed the quantization maps \cite{Ber}
\begin{align}
Q_{\hbar}:  C_c^{\infty}(\mathbb{R}^{2n}) &\to B_0(L^2(\mathbb{R}^n));\\
Q_{\hbar}(f)&=\int_{\mathbb{R}^{2n}} \frac{d^npd^nq}{(2\pi\hbar)^n} f(p,q) |\phi_{\hbar}^{(p,q)}\rangle \langle\phi_{\hbar}^{(p,q)}|,
\end{align}
 where  $\hbar\in (0,1]$; $B_0(H)$ is the $C^*$-algebra of compact operators on the Hilbert space $H=L^2(\bR^n)$,   and for each point $(p,q)\in \bR^{2n}$ the (projection) operator $|\phi_{\hbar}^{(p,q)}\rangle \langle \phi_{\hbar}^{(p,q)}|:L^2(\bR^n)\to L^2(\bR^n)$ is induced by the normalized wavefunctions, where $x\in \mathbb{R}^n$,
\begin{equation}
\phi_{\hbar}^{(p,q)}(x)=(\pi\hbar)^{-n/4}e^{-ipq/2\hbar}e^{-ipx/\hbar}e^{-(x-q)^2/2\hbar}\:, \quad \phi_{\hbar}^{(p,q)}\in L^2(\mathbb{R}^n),
\end{equation}
defining the well-known  (Schrodinger) {\em coherent states}. Inspired by Dixmier's concept of a {\em continuous bundle} \cite{Dix}, Rieffel  showed that \cite{Rie89,Rie94} 
\begin{enumerate}
\item The fibers $A_0= C_0(\bR^{2n})$ and  $A_{\hbar}= B_0(L^2(\bR^n))$, $h\in(0,1]$, can be combined into a (locally non-trivial) continuous bundle $A$ of $C^*$-algebras over $I=[0,1]$;
\item  $\tilde{A}_0=C^{\infty}_c(\bR^{2n})$ is a  dense Poisson subalgebra of $A_0$.
 \item  Each quantization map  $Q_{\hbar}:\tilde{A}_0\to A_{\hbar}$ is linear, and if we also define $Q_0: \tilde{A}_0 \hookrightarrow A_0$ as the inclusion  map, then the ensuing family $Q = (Q_\hbar)_{\hbar \in I}$
 satisfies:
  \begin{enumerate}
\item  Each map $Q_{\hbar}$ is self-adjoint, i.e.\ $Q_{\hbar}(\overline{f})=Q_{\hbar}(f)^*$ (where $f^*(x)=\overline{f(x)}$). 
\item  For each $f\in \tilde{A}_0$ the following cross-section of the bundle is continuous:
\begin{align}
    &0\to f;\\
    &\hbar\to Q_{\hbar}(f) \ \ \  (\hbar\in I\backslash\{0\})).
\end{align}
\item  Each pair $f,g\in\tilde{A}_0$ satisfies the {\bf Dirac-Groenewold-Rieffel condition}:
\begin{align}
    \lim_{\hbar\to 0}\left|\left|\frac{i}{\hbar}[Q_{\hbar}(f),Q_{\hbar}(g)]-Q_{\hbar}(\{f,g\})\right|\right|_{\hbar}=0.\label{Diracgroenewold}
\end{align} 
\end{enumerate}
\end{enumerate}
This led to the general concept of a {\em strict deformation of a Poisson manifold} $X$ \cite{Rie89,Lan98}, which we here state in the case of interest to us in which $X$ is compact, or more generally in which $X$ is a manifold with stratified boundary \cite{LMV,Pflaum}. In that case  the space $I$ in which $\hbar$ takes values cannot be all of $[0,1]$, but should be a subspace $I\subset [0,1]$ thereof that at least contains 0 as an accumulation point. This is assumed in what follows. Furthermore, the  Poisson bracket on $X$ is denoted, as usual, by $\{\cdot, \cdot\}: C^\infty(X)\times C^\infty(X) \to C^\infty(X)$ (where the smooth space $C^\infty(X)$ is suitably defined when $X$ is a more complicated object thatn a compact smooth manifold as we shall say shortly). 
 \begin{definition}\label{def:deformationq}
A {\bf strict deformation quantization} (according to \cite{Lan17}\footnote{
We stress that the some authors adopt a different notion of strict (deformation) quantization. For example, in Rieffel's approach  the same (quantum) algebra is used and the product, depending on $\hbar$, is deformed \cite{Rie89,Rie94}. In this setting, the image of the quantization map lies in an algebra with a 'new' product. We instead follow the definitions introduced by Landsman \cite{Lan98,Lan17} (who adapted Rieffel's original ideas), where for $\hbar>0$ he uses the non-commutative algebras with their intrinsic product independent of $\hbar$. The $\hbar$-dependence in turn is put in the quantization map itself. The term 'strict deformation quantization' we use in this paper is therefore related to Landsman's notion of quantization.
}) of a compact Poisson manifold $X$  consists of an index space  $I\subset [0,1]$, including $0$ as accumulation point,  for $\hbar$ as detailed above, as well as:
 \begin{itemize}
\item  A continuous
 bundle of unital $C^*$-algebras $(A_{\hbar})_{\hbar\in I}$ over $I$  with $A_0=C(X)$ equipped with the standard commutative  $C^*$-algebra structure with respect to the norm $\|\cdot\|_\infty$;
 \item A $\|\cdot\|_\infty$-dense Poisson suabalgebra  $\tilde{A}_0 \subseteq C^\infty(X) \subset A_0$ (on which $\{\cdot, \cdot\}$ is defined);
 \item  A  family $Q = (Q_\hbar)_{\hbar \in I}$ of  linear maps $Q_{\hbar}:\tilde{A}_0\to A_{\hbar}$ indexed by $\hbar\in I$ (called 
  {\bf quantization maps})  such that $Q_0$ is the inclusion map $\tilde{A}_0 \hookrightarrow A_0$, and the above
  conditions (a) - (c) hold, as well as  $Q_{\hbar}(\mathrm{1}_X)=\mathrm{1}_{A_{\hbar}}$ (the unit of $A_{\hbar}$).
\end{itemize}
\end{definition}

It follows from the definition of a continuous bundle of $C^*$-algebras that two continuity properties holds
\begin{align} \lim_{\hbar\to 0}\|Q_{\hbar}(f)\|_\hbar=\|f\|_{\infty}\end{align} and 
\begin{align} \lim_{\hbar\to 0}\|Q_{\hbar}(f)Q_{\hbar}(g)-Q_{\hbar}(fg)\|_\hbar=0\end{align}
hold automatically \cite{Dix,Lan98,Lan17}. 
\subsection{Spin systems and generalizations}
{\em Mean-field quantum spin systems}\footnote{A typical example of a mean field quantum spin system is the Curie--Weiss model (see for example \cite{ABN,CCIL,IL,Weyl,VGRL18,vandeVen} and references therein).} fit into this framework. There,  the index set $I$ is given by ($0\notin\bN:=\{1,2,3,\ldots\}$)
\begin{equation}
I= \{1/N \:|\: N \in  \mathbb{N}\} \cup\{0 \}\equiv (1/\bN) \cup\{0 \}, \label{defI}\
\end{equation}
with the  topology inherited from $[0,1]$.  That is, we  put $\hbar=1/N$, where $N\in\mathbb{N}$ is interpreted as  the number of sites of the model; our interest is the limit $N\to\infty$. In the framework of $C^*$-algebraic quantization theory, the analogy between the ``classical" limit $\hbar\to 0$ in typical examples from mechanics (see e.g. our first example \cite{Hel}) and the ``thermodynamic" limit $N\to\infty$ in typical quantum spin systems (see e.g. \cite{Lieb,LMV}) is developed in detail in \cite{Lan17}. We remark that the limit $N\to\infty$ can be taken in two entirely
different ways, which depends on the class of observables one considers, namely either {\em quasi-local} observables 
 or {\em macroscopic} observables. The former are the ones traditionally studied for quantum spin systems, but the latter relate these systems to strict deformation quantization, since macroscopic observables are precisely defined by (quasi-) symmetric sequences which form the continuous cross sections of a continuous bundle of $C^*$-algebras.  This continuous bundle of $C^*$-algebras is defined over base space $I$ given by \eqref{defI} with fibers
 \begin{align}
 A_0&=C(S(M_k(\bC))\equiv C(X_k); \label{B0}\\
 A_{1/N}&= M_k(\bC)^{\otimes N}\cong M_{k^N}(\bC) \label{BN},
   \end{align}
and continuity structure specified by continuous cross- sections which are thus given by all quasi-symmetric sequences \cite{LMV} \cite[Ch.10]{Lan17}.\footnote{The same result holds for an arbitary unital $C^*$-algebra $B$ playing the role of  the matrix algebra $M_k(\bC)$ in the above setting \cite{Lan17}.} We refer to the appendix for some useful definitions, or to \cite{LMV} for a more comprehensive explanation. The space $X_k=S(M_k(\bC))  \subset \mathbb{R}^{k^2-1}$ has the structure of a  compact Poisson manifold with {\em stratified} boundary. The space $C^\infty(X_k)$ is here made of the  restrictions to $X_k$ of the smooth functions in $\mathbb{R}^{k^2-1}$
and the Poisson bracket is the restriction ${\bf x} \in X_k$
\begin{align}\label{pbra}
\{f,g\}({\bf x}) = \sum_{a,b,c =1}^{k^2-1}C^c_{ab} x_c \frac{\partial f}{\partial x_a}\frac{\partial g}{\partial x_b}\:, \quad {\bf x}\in \mathbb{R}^{k^2-1}
\end{align}
for $f,g\in C^\infty(\mathbb{R}^{k^2-1})$   and where $C^c_{ab}$ are the structure constants of $SU(k)$ (see Sect. 2.3 of \cite{LMV} for details).
In turn, the   Poisson algebra $\tilde{A}_0$ dense in $A_0 = C(X_k)$ is made of the restrictions to $X_k$ of the polynomials 
in the $k^2-1$ coordinates of $\mathbb{R}^{k^2-1}$ 

Let us pass to describe $Q_{1/N}$. Each  polynomial $p$ of degree $L$ uniquely corresponds to a  polynomial of symmetric  elementary tensors of the form $b_{j_1}\otimes_s\cdot\cdot\cdot\otimes_sb_{j_L}$, where $ib_1,\ldots, ib_{k^2-1}$ form a basis of
the Lie algebra of $SU(k)$. That is the image of $p$ according to $Q_{1/N}$.
More precisely, if $$p_L(x_1,\ldots, x_{k^2-1}) = x_{j_1} \cdots x_{j_L}\quad \mbox{where $j_1,\ldots, j_L \in   \{1,2,\ldots, k^2-1\}$,}$$ 
the quantization maps $Q_{1/N}: \Tilde{A}_0 \to M_k(\bC)^N$ act as (see the appendix for $S_{L,N}$)
\begin{align}\label{deformationqunatizaion}
 Q_{1/N}(p_L) &=
\begin{cases}
    S_{L,N}(b_{j_1}\otimes_s\cdot\cdot\cdot\otimes_sb_{j_L}), &\ \text{if} \ N\geq L \\
    0, & \ \text{if} \ N < L,
\end{cases}\\
Q_{1/N}(1) &= \underbrace{I_k \otimes \cdots \otimes I_k}_{\scriptsize N \: times}. \label{deformationqunatizaion2},
\end{align}
and more generally they are defined as  the unique continuous and linear extensions of the written maps.\\
It has been shown in  \cite{LMV} that the quantization maps $Q_{1/N}$ satisfy all the axioms of Definition \ref{def:deformationq}.\footnote{In particular the quantization maps define (quasi)-symmetric sequences, and hence macroscopic observables.} These data  together imply the existence of a strict deformation quantization of the Poisson manifold $X_k=S(M_k(\bC))$ (see \cite[Theorem 3.4]{LMV} for a detailed proof).

We specialize these models to the case $k=2$. One-dimensional quantum spin systems arising in that way  are widely studied in (condensed matter) physics, but also in mathematical physics they form an important field of research, especially in view of spontaneous symmetry breaking (SSB). 
One tries to calculate quantities like the free energy, or the entropy of the system in question and considers their thermodynamic limit as the number of sites $N$ increas to infinity  \cite{Lieb}. For this reason the case $k=2$  is already of huge interest, since each site of such a spin chain is exactly described by the algebra of $(2\times 2)$-matrices. On the other hand the Bloch sphere $S^2$ acting as a classical phase space which describes a physical system may be a spin system of total spin $j$, but it can also be a collection of $n$ two-level atoms \cite{BZ06} corresponding to a spin chain of $n$ sites, which is for example the case for the quantum Curie-Weiss model \cite{LMV}. Inspired by that model, which admits a classical limit\footnote{This means that $\langle\Psi_N^{(0)},Q_{1/N}(f)\Psi_N^{(0)}\rangle$ admits a limit as $N\to\infty$ for any function $f\in\Tilde{A}_0$, and $\Psi_N^{(0)}$ the ground state eigenvector of the quantum CW Hamiltonian (see \cite[Theorem 4.1]{LMV} for details).} on $S^2$ (i.e. the smooth boundary of $X_2=S(M_2(\bC))\cong B^3$, where $B^3$ denotes the closed unit ball in $\mathbb{R}^3$), we asked ourselves if the quantization maps $Q_{1/N}$ quantizing  $X_2$ could in general be related to another  well-known strict deformation quantization of $S^2$ whose details are explained in what follows.\footnote{Of course, one can always try to restrict $\Tilde{A}_0$ to $\Tilde{A}_0'$ but in that case the same manifolds are quantized which is not of particular new interest.}

From the mathematical side, we observe  that $k=2$ is the unique case where $X_k$ admits a smooth boundary, as said  $X_2 = B^3$ and $\partial X_2 = S^2$.
Furthermore
$S^2$ is a Poisson submanifold of $B^3$, when the latter is equipped with the Poisson structure (\ref{pbra})  specialized to $k=2$, so that $C^a_{bc}= \epsilon_{abc}$. This is because  $S^2$ (and also $B^3$) is invariant under the flow of the Hamilton vector fields of $\mathbb{R}^{k^2-1}$ constructed out of the Poisson bracket (\ref{pbra}).
For $k=2$,  we precisely have
\begin{align}
\{f, g\}^{(B^3)}|_{S^2} = \{f|_{S^2}, g|_{S^2}\}^{(S^2)} \quad \mbox{  if $f,g \in \tilde{A}_0$},\label{restPoi}
\end{align}
with obvious notation. In particular,
\begin{align}\label{pbra}
\{f,g\}^{(B^3)}({\bf x}) = \sum_{a,b,c =1}^{3}\epsilon_{abc}x_c \frac{\partial f}{\partial x_a}\frac{\partial g}{\partial x_b}\:, \quad {\bf x}\in B^{3}.\:
\end{align}

This paper only concerns the case $k=2$. In a sense, we promote at  quantum level  the illustrated interplay of the two symplectic structures of $X_2 = B^3$  and $\partial B^3=S^2$.
 As a matter of fact, we consider the quantization elements  $Q_{1/N}(f)\in B((\mathbb{C}^2)^{N \otimes})$   under the maps \eqref{deformationqunatizaion}-\eqref{deformationqunatizaion2} referred to the symplectic structure of $B^3$. Next we
 restrict the operators $Q_{1/N}(f)$ to a  suitable common invariant subspace of  $(\mathbb{C}^2)^{N \otimes}$.
It turns out that, for large $N$,   these restricted operators correspond to the image of {\em another} quantization map $Q'_{1/N}$ acting on $C(\partial B^3)$ and referring to the symplectic structure of $\partial B^3$.

 The said invariant  subspace\footnote{This space is clearly invariant under the maps \eqref{deformationqunatizaion} - \eqref{deformationqunatizaion2}.} is $\text{Sym}^N(\mathbb{C}^2)\subset (\mathbb{C}^2)^{N \otimes}$, for which the corresponding algebras $B(\text{Sym}^N(\mathbb{C}^2))$ exactly correspond to the fibers (for $N\neq 0$) of another continuous bundle of $C^*$-algebras given by \eqref{B02} - \eqref{BN2} below. It is a well-known fact that these fibers together with quantization maps \eqref{defquan3} - \eqref{defquan4} below give rise to a strict deformation quantization of $S^2$ \cite{BMS94,Lan98,Pe72} according to Definition \ref{def:deformationq}.
%

Indicating the algebra of bounded operators by $B(\text{Sym}^N(\mathbb{C}^2))$, it is known \cite[Theorem 8.1]{Lan17} that
 \begin{align}
 A_0'&=C(S^2); \label{B02}\\
 A_{1/N}'&=M_{N+1}(\mathbb{C})\cong B(\text{Sym}^N(\mathbb{C}^2))\label{BN2},
   \end{align}
are the fibers of a continuous bundle of $C^*$-algebras over the same base space  $I$ as in  \eqref{defI} whose continuous cross-sections are given by all sequences $(a_{1/N})_{N\in\mathbb{N}}\in\Pi_{n\in\mathbb{N}}A_{1/N}'$ for which $a_0\in C(S^2)$ and $a_{1/N}\in M_{N+1}(\mathbb{C})$  and  such that the
sequence  $(a_{1/N})_{N\in\mathbb{N}}$ is asymptotically equivalent to $(Q_{1/N}'(a_0))_{N\in\mathbb{N}}$, in the sense that
\begin{align}
\lim_{N\to\infty}||a_{1/N}-Q_{1/N}'(a_0)||_N=0.\label{equivalencebookklaas}
\end{align}
Here, the symbol $Q_{1/N}'$ denotes the quantization maps
\begin{align}
Q_{1/N}':\tilde{A}_0' \to A_{1/N}',
\end{align}
where $\tilde{A}_0'\subset C^{\infty}(S^2)\subset A'_0$ is the dense Poisson subalgebra made of polynomials in three real variables restricted the sphere $S^2$. The maps $Q_{1/N}'$ are defined by\footnote{Equivalent definitions of these quantization maps are used in literature, see e.g. \cite{Lan17,Pe72}.} the integral computed in weak sense
\begin{align}\label{defquan3}
Q_{1/N}'(p)& :=
 \frac{N+1}{4\pi}\int_{S^2}p(\Omega)|\Omega\rangle\langle\Omega|_Nd\Omega\:,
\end{align}
where $p$ denotes an arbitrary polynomial restricted to  $S^2$, $d\Omega$ indicates the unique $SO(3)$-invariant Haar measure on ${S}^2$ with $\int_{{S}^2} d\Omega = 4\pi$, and $|\Omega\rangle\langle\Omega|_N\in B(\text{Sym}^N(\mathbb{C}^2))$ are so-called $N$ coherent spin states
defined in Appendix \ref{appB}. In particular, if $1$ is the constant function $1(\Omega)=1, \ (\Omega\in S^2)$, and $1_N$ is the identity on $A_{1/N}'=B(\text{Sym}^N(\mathbb{C}^2))$, the provious definition implies 
\begin{align}
Q'_{1/N}(1) = 1_N  \label{defquan4}\:.
\end{align}
Indeed, it can be shown that the quantization maps \eqref{defquan3} - \eqref{defquan4} satisfy the axioms of Definition \ref{def:deformationq}, which implies the existence of a strict deformation quantization of $S^2$.\footnote{We remark that $S^2$ is a special case of a regular integral coadjoint orbit in the dual of the Lie algebra associated to $SU(2)$, which can be identified with $\mathbb{R}^3$. In fact, this theory can be generalized to arbitrary compact connected Lie groups \cite{Lan98}.}  These quantization maps,  constructed from a family coherent states (as opposed to the maps \eqref{deformationqunatizaion} - \eqref{deformationqunatizaion2} which are defined in a complete different way), also define a so-called {\bf Berezin quantization} \cite{Lan98} for which \eqref{Berezinprop} typically holds as well as positivity, in that $Q_{1/N}(f)\geq 0$ if $f\geq 0$ almost everywhere on $S^2$. 

The main result of this work is an asymptotic  relation connecting the bulk and the boundary quantization maps:
$$\left|\left|Q_{1/N}(p)|_{Sym^{(N)}( (\bC^2)^{N\otimes})}  - Q'_{1/N}(p|_{S^2})\right|\right|_N \to 0 \quad \mbox{for $N\to +\infty$}
\:, \quad p \in \tilde{A}_0$$
established in Theorem \ref{MAIN}. We stress that  the validity of the {\em  Dirac-Groenewold-Rieffel condition} (\ref{Diracgroenewold}) for both maps is possible just thank to (\ref{restPoi}).
\newline
\newline
The plan of this paper is as follows. In section 2 we state and prove our main theorem (Theorem \ref{MAIN}) establishing a connection between the strict deformation quantization of $X_2$ and the one of $S^2$ defined above. We show that the quantization maps $Q_{1/N}$ defined by \eqref{deformationqunatizaion} - \eqref{deformationqunatizaion2} whose images are restricted to $\text{Sym}^N(\mathbb{C}^2)$ satisfy the identity above with respect to the other quantization map $Q'_{1/N}$. In section 3 we apply our theorem to the Curie-Weiss model which links the corresponding quantum Hamiltonian to its classical counter part on the sphere. In the appendix we provide a comprehensive overview of useful definitions.

\section{Interplay  of bulk quantization map $Q_{1/N}$ and  boundary quantization map $Q'_{1/N}$}
In order to arrive at the main thereom of this paper we first introduce some vector spaces. We let  $P_N$ to be  the complex vector space of polynomials in the  variables $x,y,z \in \bR^3$  of degree $\leq N$ where  $N\geq 1$, and let  $P_N(S^2)$ be the vector space made of the restrictions to $S^2$ of those polynomials.

\subsection{Preparatory results on $Q'_{1/N}$ and harmonic polynomials}
Definition (\ref{defquan3}) can actually be stated replacing the polynomial $p$ by a generic  $f\in C(S^2)$, though its meaning as a quantization map is valid for  the domain of the  polynomials restricted to  $S^2$ as indicated in (\ref{defquan3}).
The  map associating  $f \in C(S^2)$  with 
\begin{align}
&Q'_{1/N}(f) : \text{Sym}^N(\mathbb{C}^2) \to \text{Sym}^N(\mathbb{C}^2);\\
&Q'_{1/N}(f) := \frac{N+1}{4\pi}\int_{S^2} f(\Omega) |\Omega\rangle \langle \Omega|_N d\Omega,\label{one}
\end{align}
is well-defined and  {\em it is surjective on $B(\text{Sym}^N(\mathbb{C}^2))$} since,
 for every  $A : \text{Sym}^N(\mathbb{C}^2)\to \text{Sym}^N(\mathbb{C}^2)$, there exists a function
 $p \in P_N(S^2)$  such that $A=Q'_{1/N}(p)$.  Indeed, that the function
\beq p(\Omega) := tr(A \Delta_N^{(1)}(\Omega))\:,\label{defp}\eeq
where  $\Omega \in S^2$ and $\Omega\mapsto\Delta_N^{(1)}\in\text{Sym}^N(\mathbb{C}^2)$ is defined by Definition (2.6) in \cite{Germanpaper}, defines a polynomial on the sphere, i.e. 
\beq tr(A \Delta_N^{(1)}) \in P_N(S^2).\label{inPN}\eeq
In particular,  we realize that the {\em linear} map (\ref{one}) {\em cannot be injective on the domain  $C(S^2)$} since this space is  infinite dimensional whereas the co-domain is finite dimensional.
Nevertheless, if restricting the domain to  $P_N(S^2)$, the said map turns out to be bijective.
\begin{proposition}\label{prop1}
 The map 
\beq P_N(S^2) \ni p \mapsto Q'_{1/N}(p) := \frac{N+1}{4\pi}\int_{S^2} p(\Omega) |\Omega\rangle \langle \Omega|_N d\Omega \in B(\text{Sym}^N(\mathbb{C}^2))\label{one1}\eeq
is a bijection for $N>1$.
\end{proposition}

\begin{proof}
The said map is obviously surjective, as already observed, because, by defining  $p(\Omega) := tr(A \Delta_N^{(1)}(\Omega))$ for 
$A\in B(Sym^{(N)}( \bC^2) )$, we have  $A = Q'_{1/N}(p)$. Let us prove injectivity. It is well known \cite{Bookpol} that $\dim(P_N(S^2))= (N+1)^2$ if $N>1$. On the other hand, $\dim(B(Sym^{(N)}( \bC^2) ))= (N+1)^2$ as one immediately proves. As
$\dim(B(Sym^{(N)}( \bC^2) ))= \dim(P_N(S^2)) < +\infty$, surjectivity implies injectivity from elementary results of linear algebra.
\end{proof}
\noindent
%

\noindent 
Going back to Weyl, let us recall a few results on the theory of $SO(3)$ representations of polynomials restricted to the unit sphere.
 The group $SO(3)$ admits a natural representation on $P_N(S^2)$ given by
\beq
SO(3) \ni \bR \mapsto \rho_R\:,  \quad  (\rho_R p)(\Omega) := p(R^{-1}\Omega)\quad\forall p\in P_N(S^2)\:, \forall \Omega \in S^2. \label{rho}
\eeq 
In turn, the space $P_N(S^2)$ admits a direct decomposition into {\em  invariant and irreducible} subspaces under the action of $\rho$, viz.
$$P_N(S^2)= \bigoplus_{j=0,1, \ldots, N} P_N^{(j)}(S^2).$$
Each subspace $P_N^{(j)}(S^2)$ consists of the restrictions to $S^2$  of the   homegeneous polynomials  of order $j$ that  are also harmonic functions. $P_N^{(j)}(S^2)$ has dimension $2j+1$.  

\begin{example}
{\em  If $N=2$
$$P_2(S^2)=P_2^{(0)}(S^2) \oplus P_2^{(1)}(S^2) \oplus P_2^{(2)}(S^2)\:.$$
In the right-hand side,  the first subaspace is the span 
of the restriction to $S^2$ of the constant polynomial  $p(x,y,z):=1$, the second one is the span 
of the restrictions of  the three polynomials $p_j(x,y,z):= x_j$, $j=1,2,3$ where $x_1=x,x_2=y,x_3=z$,  and the  third one is the  the span 
of the restrictions to $S^2$ of five elements suitably chosen\footnote{The restrictions to $S^2$ of these six polynomials and the one of  the above $p$ form a linearly dependent set.}  of the six polynomials $p_{ij}(x,y,z) := x_ix_j - \frac{1}{3}\delta_{ij}(x^2+y^2+z^2)$ for
$i,j \in \{1,2,3\}$.}
\end{example}

If  $\rho_R^{(j)}$ is the restriction of $\rho_R$ to $P_N^{(j)}(S^2)$ and $\{p^{(j)}_{m}\}_{m=-j, -j+1, \ldots, j-1, j}$ is a basis of 
$P_N^{(j)}(S^2)$, we find
\beq \rho^{(j)}_R p^{(j)}_{m} = \sum_{m' =-j}^jD^{(j)}_{mm'}(R^{-1}) p^{(j)}_{m'}\label{rho2} \:.\eeq
Each class of matrices $\{D^{(j)}(R)\}_{R\in SO(3)}$ defines an irreducible representation of $SO(3)$ in $\bC^{2j+1}$.  These representations  are completely fixed by their dimension 
i.e., by $j$, up to equivalence given by similarity transformations, and different $j$ correspond to similarity inequivalent representations. 
Every irreducible representation of $SO(3)$ is unitarily equivalent to one of the $D^{(j)}$.

\subsection{The main theorem}
Before arriving at the main theorem of this paper, we recall that by construction the space $\tilde{A}_0$ is the complex vector space of polynomials in three variables on the closed unit ball $B^3$ which in particular contains all polynomials of $P_M \ (M\in\mathbb{N})$ restricted to $B^3$  \cite{LMV}. In the proof of the theorem we occasionally use the space $\tilde{A}_0$ as well as $P_N$, where the former is the domain of the quantization maps $Q_{1/N}$, whereas the latter is used to underline the degree of the polynomial in question.
\begin{theorem}\label{MAIN}
If $p\in \tilde{A}_0$, then
$$\left|\left|Q_{1/N}(p)|_{Sym^{(N)}( (\bC^2)^{N\otimes})}  - Q'_{1/N}(p|_{S^2})\right|\right|_N \to 0 \quad \mbox{for $N\to +\infty$}
\:,$$
the (operator) norm being the one on B($\text{Sym}^N(\mathbb{C}^2))$,
\end{theorem}

\begin{remark}
{\em We stress that the result does not automatically imply  that the cross-sections \eqref{deformationqunatizaion} - \eqref{deformationqunatizaion2} whose images are resticted to $\text{Sym}^N(\mathbb{C}^2)$ are also continous cross-sections of the fibers defined in \eqref{B02} - \eqref{BN2}, since $f\in A_0=C(B^3)$ does {\em not} imply that $f\in A_0'=C(S^2)$.}
\end{remark}

\begin{proof} We start the proof by  discussing the interplay between the action of $SO(3)$ and the quantization maps $Q_{1/N}$, defined in \eqref{deformationqunatizaion}. We first focus on a homogeneous polynomial of order $M<N$.\footnote{As we are dealing with a limit in $N$, we can safely take $N$ such that $M<N$.}
If   $k_1,\ldots, k_M$ are taken in  $\{1,2,3\}$ and
\beq p_{k_1\cdots k_M}(x,y,z) :=x_{k_1} \cdots x_{k_M}\:,
\label{p}\eeq
the representation (\ref{rho}) implies that  
\beq \left(\rho_R p_{k_1\cdots k_M}\right)(x,y,z)   = {(R_U^{-1})_{k_1}}^{j_1}  \cdots {(R_U^{-1})_{k_M}}^{j_M} p_{k_1\cdots k_M}(x,y,z)\label{rest1}\:.\eeq
We stress that when restricting to $S^2$, every $p^{(j)}_m$ is a linear combination of  the restricitions of the polynomials $p_{k_1\cdots k_M}$ so that, by  extending (\ref{rest1}) by linearity and working on $p^{(j)}_m$, (\ref{rest1}) must coincide with (\ref{rho2})
$$
\left(\rho_R^{(j)} p^{(j)}_m\right)(x,y,z) = \sum_{m'=-j}^j D^{(j)}_{mm'}(R^{-1}) p^{(j)}_{m'}(x,y,z) \:, \quad (x,y,z) \in S^2\: \label{identityonsphere}
$$
Since both sides are restrictions of homegeneous polynomials of the same degree $j$, this identity is valid also removing the constraint
$(x,y,z) \in S^2$:
\beq 
\left(\rho_R^{(j)} p^{(j)}_m\right)(x,y,z) = \sum_{m'=-j}^j D^{(j)}_{mm'}(R^{-1}) p^{(j)}_{m'}(x,y,z) \:, \quad (x,y,z) \in \bR^3\:.\label{aggno}
\eeq
where now the $p^{(j)}_{m}$ are homegeneous polynomials in $P_M$ whose restrictions are the basis elements of $P^{(j)}_M(S^2)$ with the same name. We underline that for our quantization maps $Q_{1/N}$ we  need $p^{(j)}_{m}$  to be a polynomial in $\tilde{A}_0$, rather than in $P_M$.  However, since $\tilde{A}_0$ contains all polynomials of $P_M$ restricted to $B^3$ which has non-empty interior, polynomials of $P_M$ are in one-to-one correspondence with those of  $\tilde{A}_0$. Therefore, in view of \eqref{aggno} the same statement holds when we replace $(x,y,z)\in\mathbb{R}^3$ by $(x,y,z)\in B^3$.
Now, by definition of the quantization maps  $Q_{1/N}$ we know that
$$Q_{1/N}(p_{k_1\cdots k_M}) = S_{N,M}\left(\sigma_{k_1}  \otimes \cdots \otimes\sigma_{k_M}  \otimes\underbrace{I \otimes \cdots \otimes I}_{N-M\: times}\right)\:.$$ 
Let us indicate by $R_U\in SO(3)$ the image of  $U\in SU(2)$ through the universal covering homomorphism $\Pi: SU(2) \to SO(3)$.
This covering homomorphism as is well known satisfies (using the summation convention on repeated indices)
\beq U\sigma_j U^* = {(R^{-1}_U)_j}^k \sigma_k.\label{rot} \eeq
Remembering that  $\text{Sym}^N(\mathbb{C}^2)$ is invariant under the tensor representation
$\underbrace{U\otimes \cdots \otimes U}_{N\: times}$, we have
$$\underbrace{U\otimes \cdots \otimes U}_{N\: times}|_{\text{Sym}^N(\mathbb{C}^2)} \:Q_{1/N}(p_{k_1\cdots k_M})|_{\text{Sym}^N(\mathbb{C}^2)}\: \underbrace{U^*\otimes \cdots \otimes U^*}_{N\: times}|_{\text{Sym}^N(\mathbb{C}^2)}$$
$$= \left(\underbrace{U\otimes \cdots \otimes U}_{N\: times} Q_{1/N}(p_{k_1\cdots k_M})\underbrace{U^*\otimes \cdots \otimes U^*}_{N\: times}\right)|_{\text{Sym}^N(\mathbb{C}^2)} $$
$$= {(R^{-1}_U)_{k_1}}^{j_1}  \cdots {(R^{-1}_U)_{k_M}}^{j_M}   Q_{1/N}(p_{j_1\ldots j_M})|_{\text{Sym}^N(\mathbb{C}^2)} \:.$$
Let us consider linear combinations $p^{(j)}_m$ of polynomials $p_{k_1\cdots k_M}$ whose restriction to $S^2$ define the basis element, indicated with the same symbol, 
 $p^{(j)}_m \in P_M^{(j)}(S^2)$. Since  the map $Q_{1/N}$ is linear, 
from (\ref{aggno}) we have
$$\left(\underbrace{U\otimes \cdots \otimes U}_{N\: times} Q_{1/N}(p^{(j)}_m)\underbrace{U^*\otimes \cdots \otimes U^*}_{N\: times}\right)|_{\text{Sym}^N(\mathbb{C}^2)}$$ \beq = \sum_{m'} D^{(j)}_{mm'}(R^{-1}) Q_{1/N}(p^{(j)}_{m'})|_{\text{Sym}^N(\mathbb{C}^2)} \label{fund}\:.\eeq
Let us now pass to the other quantization map $Q'_{1/N}$ observing that \eqref{fund} and Proposition \ref{prop1} entail
\beq Q_{1/N}(p^{(j)}_m)|_{\text{Sym}^N(\mathbb{C}^2)} =Q'_{1/N}(q^{(j)}_m)= \frac{N+1}{4\pi} \int_{S^2} q^{(j)}_m(\Omega) |\Omega \rangle \langle\Omega|_N d\Omega\label{two}\eeq
for some  $q^{(j)}_m \in P_N(S^2)$ (where $N>M$ in general)  is the unknown restriction to  $S^2$ of a polynomial in $P_N$.
Exploiting (\ref{fund}) and linearity we find 
$$\underbrace{U\otimes \cdots \otimes U}_{N\: times}|_{\text{Sym}^N(\mathbb{C}^2)} \:Q_{1/N}(p_m^{(j)})|_{\text{Sym}^N(\mathbb{C}^2)}\: \underbrace{U^*\otimes \cdots \otimes U^*}_{N\: times}|_{\text{Sym}^N(\mathbb{C}^2)}$$
\beq =  \frac{N+1}{4\pi} \int_{S^2} \sum_{m'} D^{(j)}_{m'}(R^{-1})
 q^{(j)}_{m'}(\Omega) |\Omega\rangle \langle \Omega|_N d\Omega\label{four}\:.\eeq
Again, from  (\ref{one}) we have the general relation
$$VA^{(N)}_f V^*= \frac{N+1}{4\pi}\int_{S^2} f(\Omega) V|\Omega\rangle \langle \Omega|_NV^*d\Omega\:.$$
Specializing to $V= \underbrace{U\otimes \cdots \otimes U}_{N\: times}|_{\text{Sym}^N(\mathbb{C}^2)} $ we obtain (see Lemma \ref{lemma} below)
\beq V|\Omega \rangle = e^{i\alpha_{U,\Omega}} |R_U\Omega\rangle\label{provare}\eeq
where the phase is irrelevant as it disappears in view of later computations,
hence
$$V A^{(N)}_f V^*= \frac{N+1}{4\pi}\int_{S^2} f(\Omega) |R_U\Omega\rangle \langle R_U\Omega|_Nd\Omega$$
$$=  \frac{N+1}{4\pi}\int_{S^2} f(R^{-1}_U\Omega) |R_UR^{-1}_U\Omega\rangle \langle R_UR_U^{-1}\Omega|_NdR_U^{-1}\Omega,$$
namely
\beq \underbrace{U\otimes \cdots \otimes U}_{N\: times}|_{\text{Sym}^N(\mathbb{C}^2)}  \: A \underbrace{U^*\otimes \cdots \otimes U^*}_{N\: times}|_{\text{Sym}^N(\mathbb{C}^2)} = \frac{N+1}{4\pi}\int_{S^2} f(R^{-1}_U\Omega) |\Omega\rangle \langle \Omega|_Nd\Omega\label{three}\eeq
where we took advantage of  $d\Omega = dR^{-1}\Omega$ if  $R\in SO(3)$.
To conclude, if $A = Q_{1/N}(p^{(j)}_m)$, identity    (\ref{four}) yields
$$\int_{S^2} \sum_{m'}D^{(j)}_{mm'} (R^{-1}) q^{(j)}_{m'}(\Omega)|\Omega\rangle \langle \Omega|_N d\Omega=\int_{S^2} q^{(j)}_m (R^{-1}_U\Omega) |\Omega\rangle \langle \Omega|_Nd\Omega\:.$$
Since the map  (\ref{one}) is bijective on $P_N(S^2)$ it must be
\beq 
q^{(j)}_m (R^{-1}\Omega)=\sum_{m'}D^{(j)}_{mm'} (R^{-1}) q^{(j)}_{m'}(\Omega)\:, \quad \forall \Omega \in S^2 \:,\forall R \in SO(3)\label{pol1}\eeq
Linearity and bijectivity of the map (\ref{one1}) also implies that, varying $m=-j,-j+1,\ldots, j-1,j$ the functions 
$q^{(j)}_m$ form a basis of a $2j+1$ dimensional subspace of $P_N(S^2)$.
We can expand  each of these functions over the basis of functions $p^{(j)}_m$ of $P_N(S^2)$:
\beq 
q^{(j')}_{m'} = \sum_{j=0}^N \sum_{m=-j}^j C^{(j',j)}_{m'm} p^{(j)}_m\:,\label{decqp}
\eeq
where both sides are now and henceforth evaluated on $S^2$.
Here (\ref{rho2}) and (\ref{pol1}) together imply 
$$
  \sum_{j,m, \kappa} D^{(j)}_{m'\kappa}(R) C^{(j',j)}_{\kappa m} p^{(j)}_m = \sum_{j,m,\ell} C^{(j',j)}_{m'm}
 D^{(j)}_{m \ell}(R)p^{(j)}_\ell \:, 
$$
that is
$$
  \sum_{j,\ell , \kappa} D^{(j)}_{m'\kappa}(R) C^{(j',j)}_{\kappa \ell} p^{(j)}_\ell  = \sum_{j,m,\ell} C^{(j',j)}_{m'm} D^{(j)}_{m \ell }(R)p^{(j)}_\ell\:.
$$
Since the set of the (restrictions of to the sphere of the) $p^{(j)}$ is a basis,

$$
  \sum_{m} D^{(j)}_{m'm}(R) C^{(j',j)}_{m \ell}   = \sum_{m} C^{(j',j)}_{m'm} D^{(j)}_{m \ell }(R)\:.
$$
Since the representation $D^{(j)}$ is irreducible, Schur's lemma implies that  there are complex numbers $C^{(j',j)} $ such that 
$$ C^{(j',j)}_{m \ell}  =C^{(j',j)} \delta _{m \ell}\:.$$
In summary, (\ref{decqp2}) reduces to
\beq 
q^{(j')}_{m} = \sum_{j=0}^M C^{(j',j)} p^{(j)}_m|_{S^2}\:.\label{decqp2}
\eeq
However, since the elements in the left-hand side are $2j'+1$ whereas, for every $j$ in the right-hand side we have 
$2j+1$ elements and the spaces of these representations transform separately, the only possibility is that $ C^{(j',j)}= 0$ if $j\neq j'$. In other words,
\beq 
q^{(j, N)}_{m} =  C^{(j)}_N p^{(j)}_m|_{S^2} \quad \mbox{for every given $j=0,1,2,\ldots,M$} \:.\label{decqp3}
\eeq
where,
\begin{itemize}
\item [(i)] we have terminated $j$ to $M<N$ because the initial polynomial $p^{(j)}_m$ has been chosen in $P_M^{(j)}(S^2)$;

\item [(ii)]   we have restored the presence of $N$, since $C^{(j)}_N$  may depend on $N$.
\end{itemize}
Let us examine what happens to $ C^{(j)}_N $ at large $N$. First observe that (\ref{decqp3}) immediately implies 
$$Q_{1/N}(p^{(j)}_m)|_{\text{Sym}^N(\mathbb{C}^2)}  =  C^{(j)}_N  \frac{N+1}{4\pi} \int_{S^2}p^{(j)}_{m}(\Omega) |\Omega\rangle \langle\Omega|_N d\Omega.$$
Taking the expectation value $\langle \Omega'| \cdot |\Omega'\rangle$, we find
\beq p^{(j)}_m(\Omega') =  C^{(j)}_N \int_{S^2}p^{(j)}_m(\Omega) \frac{N+1}{4\pi} |\langle \Omega'|\Omega\rangle_N|^2d\Omega\:.\label{pp}\eeq
In Lemma \ref{Valter} below we prove that $\lim_{N\to +\infty}C^{(j)}_N$ exists and is finite.
Hence,  
$$ p^{(j)}_m(\Omega') =  \left(\lim_{N\to +\infty}C^{(j)}_N \right)   p^{(j)}_m(\Omega'),$$
where we exploited  Proposition 4.2 of \cite{LMV}, so that
$$ \lim_{N\to +\infty}C^{(j)}_N=1\:.$$
This reasonig implies the claim for the considered special polynomials since,  for $N\to +\infty$,
$$\left|\left|Q_{1/N}(p^{(j)}_m)|_{\text{Sym}^N(\mathbb{C}^2)}  -   \frac{N+1}{4\pi} \int_{S^2}p^{(j)}_m|_{S^2}(\Omega) |\Omega\rangle \langle\Omega|_N d\Omega\right|\right|_N$$
\beq = |C_N-1| \left|\left|\frac{N+1}{4\pi} \int_{S^2}p^{(j)}_m|_{S^2}(\Omega) |\Omega\rangle \langle\Omega|_N d\Omega\right|\right|_N \leq |C_N-1| ||p^{(j)}_m|_{S^2}||_\infty\to 0 \label{fine}\eeq
The found result  immediately extends to every polynomial of given degree $M$ which can be written as a linear combination of the  $p^{(j)}_m$ viewed as polynomials. 
To pass to a generic polynomal in $\tilde{A}_0$ (say of degree $M$) we observe that, as a consequence of known results  \cite{Bookpol}, the map
$$\tilde{A}_0 \ni p \mpasto p|_{S^2} \in P_M(S^2)$$
has a kernel made of all possible polynomials of the form $q(x,y,z) (x^2+y^2+z^2-1)$ with $q \in P_{M-2}$.  
Furthermore, Proposition \ref{Chris} below proves that, for every $q \in P_{M-2}$,
\beq ||Q_{1/N}(q(x,y,z) (x^2+y^2+z^2-1))|_{\text{Sym}^N(\mathbb{C}^2)}||_N \to 0 \quad \mbox{as $N\to +\infty$}.\label{final}\eeq
So, if  $p \in\tilde{A}_0$ is a polynomial of degree $M$, then  we can write for a finite number of coefficients $ C^{(j,m)}$ and some polynomial $q\in P_{M-2}$,
\begin{align}
p = \sum_{j,m} C^{(j,m)}p_m^{(j)} + q(x,y,z) (x^2+y^2+z^2-1), \label{isthistrue}
\end{align}
where the $p_m^{(j)}$ and $q$ are here intepreted as elements of $P_M$ and $P_{M-2}$ respectively,  restricted to $B^3$.
Hence,
$$Q_{1/N}(p)|_{\text{Sym}^N(\mathbb{C}^2)} = \sum_{j,m} C^{(j,m)}Q_{1/N}(p_m^{(j)})|_{\text{Sym}^N(\mathbb{C}^2)}$$
$$+ Q_{1/N}(q(x,y,z) (x^2+y^2+z^2-1))|_{\text{Sym}^N(\mathbb{C}^2)}\:.$$
The former  term on the right-hand side tends to $Q'_{1/N}(p|_{S^2})$, the latter vanishes as $N\to +\infty$ proving the thesis.
\end{proof}

\subsection{Subsidiary technical results}
\begin{lemma}\label{Valter}  
$\lim_{N\to +\infty}C^{(j)}_N$ exists and is finite. \end{lemma}
\begin{proof}
Since the  left-hand side of  (\ref{pp}) does not depend on  $N$ and the integral in the right-hand side tends to $p^{(j)}_m(\Omega')$, the only possibility that the limit $\lim_{N\to +\infty}C^{(j)}_N $ prevents from existing (or that makes it infinite)
is $p_m^{(j)}(\Omega') =0$.  This result should be true for all $\Omega'$, since  $\lim_{N\to +\infty}C_N$  is independent of 
 $\Omega'$. However the polynomial  $p_m^{(j)}$ (restricted to $S^2$) is not the zero funtion since it is an element of a basis.
\end{proof}

\begin{lemma}\label{lemma} Eq. (\ref{provare}) is true.
\end{lemma}
\begin{proof} As is well known (see   \cite{LMV} for a summary of those properties and technical references),
$$\Omega \cdot \sigma |\Omega\rangle_1 = |\Omega\rangle_1\:.$$
Applying  $U$ to  both sides gives
$$\Omega \cdot U  \sigma U^* U |\Omega\rangle_1 = U |\Omega\rangle_1.$$
Namely, from (\ref{rot}) we obtain
$$\Omega \cdot (R^{-1}_U \sigma) U  |\Omega\rangle_1 = U |\Omega\rangle_1,$$
that is
$$(R_U\Omega) \cdot \sigma \: Ur  |\Omega\rangle_1 = U |\Omega\rangle_1\:.$$
We also know that
$$(R_U\Omega) \cdot \sigma\:  |R_U\Omega\rangle_1 =  |R_U\Omega\rangle_1\:.$$
Since the eigenspace of  $(R^{-1} \Omega)\cdot \sigma$ with eigenvalue $1$ is one-dimensional, for some real $\beta_{\Omega, U}$, we must have
$$U  |\Omega\rangle_1 = e^{i\beta_{\Omega, U}}|R_U\Omega\rangle_1\:.$$
Taking advantage of $|\Omega\rangle_N = \underbrace{|\Omega\rangle_1 \otimes \cdots \otimes |\Omega\rangle_1}_{N\: times}$ and 
$V= \underbrace{U\otimes \cdots \otimes U}_{N\: times}|_{\text{Sym}^N(\mathbb{C}^2)},$ we  immedialtey achieve  (\ref{provare}) with  $\alpha_{\Omega,U}= N\beta_{\Omega,U}$.
\end{proof}

\begin{proposition}\label{Chris} Eq. (\ref{final}) is true.
\end{proposition}
\begin{proof}
We use the canonical (Dicke) basis \cite{Pe72,LMV} $|n,N-n\rangle $ for $\text{Sym}^N(\mathbb{C}^2)$ $(n=0,...,N)$ and first show that the matrix elements with respect to this basis are zero:
\begin{align}
\langle n| Q_{1/N}(q(x,y,z))Q_{1/N}(x^2+y^2+z^2-1)| k\rangle=0, \ \ (k,n=0,...,N).
\end{align}
Consider now a basis vector $|k,N-k\rangle$. We first expand $|k,N-k\rangle$ in the standard basis vectors $\beta_i$ $(i=1,...,2^N$) spanning the Hilbert space $\bigotimes^N\mathbb{C}^2$. 
We denote by $\mathscr{O}^{k}$ the orbit consisting of ${N}\choose{k}$-basis vectors $\beta_i$ with the same number of occurrence of the vectors $e_2$ and $e_1$, the two basis vectors of $\mathbb{C}^2$. By convention, we take $e_1$ such that $\sigma_3 e_1=e_1$, and $\sigma_3 e_2=-e_2$. It is not difficult to show that \cite{vandeVen,VGRL18}
\begin{align*}
|k,N-k\rangle=\frac{1}{\sqrt{{{N}\choose{k}}}}\sum_{l=1}^{{{N}\choose{k}}}\beta_{k,l}
\end{align*}
where the subindex $l$ in  $\beta_{k,l}$ labels the basis vector $\beta_{k,l}\in\beta$ within the same orbit $\mathscr{O}^{k}$. Since we have ${N}\choose{k}$ such vectors per orbit, the sum in the above equation indeed is from $l=1,...,{{N}\choose{k}}$.
By definition $Q_{1/N}(x_i^2)=S_{2,N}(\sigma_i\otimes\sigma_i)$ for $i=1,2,3$. Using a combinatorial argument and the fact that all $|k\rangle$ are symmetric it follows that
\begin{align*}
&S_{2,N}(\sigma_2\otimes\sigma_2)|k\rangle=\frac{1}{\sqrt{{{N}\choose{k}}}}\sum_{l=1}^{{{N}\choose{k}}}(\sigma_2\otimes\sigma_2\otimes 1\cdot\cdot\cdot\otimes 1)  \beta_{k,l}=&\\ &\frac{1}{\sqrt{{{N}\choose{k}}}}\bigg{(}-{{N-2}\choose{k-2}}\beta_{k-2,l}+2{{N-2}\choose{k-1}}\beta_{k,l}-{{N-2}\choose{k}}\beta_{k+2,l}.
\end{align*}
Similarly,
\begin{align*}
&S_{2,N}(\sigma_1\otimes\sigma_1)|k\rangle=&\\ &\frac{1}{\sqrt{{{N}\choose{k}}}}\bigg{(}{{N-2}\choose{k-2}}\beta_{k-2,l}+2{{N-2}\choose{k-1}}\beta_{k,l}+{{N-2}\choose{k}}\beta_{k+2,l};
\end{align*}
and
\begin{align*}
&S_{2,N}(\sigma_3\otimes\sigma_3)|k\rangle=&\\ &\frac{1}{\sqrt{{{N}\choose{k}}}}\bigg{(}{{N-2}\choose{k-2}}\beta_{k,l}-2{{N-2}\choose{k-1}}\beta_{k,l}+{{N-2}\choose{k}}\beta_{k,l}.
\end{align*}
In view of Definition \ref{def:deformationq} (property $3(b)$) the cross-section $0\to f$ and $1/N\to Q_{1/N}(f)$ defines a continuous section of the bundle implying that the following condition (see also the remark below Definition \ref{def:deformationq}) is automatically satisfied:
\begin{align}
\lim_{N\to\infty}||Q_{1/N}(f)Q_{1/N}(f)-Q_{1/N}(fg)||_N=0. \label{Klaas}
\end{align}
We apply this with $f=q(x,y,z)$ and $g(x,y,z)=x^2+y^2+z^2-1$. We first show that
\begin{align*}
 \langle n|Q_{1/N}(q(x,y,z))Q_{1/N}(x^2+y^2+z^2-1))| k\rangle=0,
\end{align*}
for all basis vectors $|n\rangle$ and $|k\rangle$ in $\text{Sym}^N(\mathbb{C}^2)$.
Indeed,  using the above identities one finds
\begin{align*}
&\langle n| Q_{1/N}(q(x,y,z)) Q_{1/N}(x^2+y^2+z^2-1) |k\rangle=&\\&
\frac{1}{\sqrt{{{N}\choose{n}}}}\frac{1}{\sqrt{{{N}\choose{k}}}}\sum_{l=1}^{{{N}\choose{n}}}\sum_{r=1}^{{{N}\choose{k}}}\langle\beta_{n,l},Q_N(q(x,y,z))\bigg{(}S_{2,N}(\sigma_1\otimes\sigma_1)+S_{2,N}(\sigma_2\otimes\sigma_2)+S_{2,N}(\sigma_3\otimes\sigma_3)\bigg{)}\beta_{k,r}\rangle -&\\&\frac{1}{\sqrt{{{N}\choose{n}}}}\frac{1}{\sqrt{{{N}\choose{k}}}}\sum_{l=1}^{{{N}\choose{n}}}\sum_{r=1}^{{{N}\choose{k}}}\langle\beta_{n,l},Q_{1/N}(q(x,y,z))\beta_{k,r}\rangle=&\\&
\frac{1}{\sqrt{{{N}\choose{n}}}}\frac{1}{\sqrt{{{N}\choose{k}}}}\sum_{l=1}^{{{N}\choose{n}}}\langle\beta_{n,l},Q_N(q(x,y,z))\bigg{(}{{N-2}\choose{k-2}}+{{N-2}\choose{k}}+2{{N-2}\choose{k-1}}-{{N}\choose{k}}\bigg{)}\beta_{k,r}\rangle
&\\&
\frac{1}{\sqrt{{{N}\choose{n}}}}\frac{1}{\sqrt{{{N}\choose{k}}}}\sum_{l=1}^{{{N}\choose{n}}}\langle\beta_{n,l},Q_N(q(x,y,z))\bigg{(}{{N}\choose{k}}-{{N}\choose{k}}\bigg{)}\beta_{k,r}\rangle=0.
\end{align*}
Since this holds for all basis vectors and $\text{Sym}^N(\mathbb{C}^2)$ is invariant under $Q_{1/N}(q(x,y,z))$ and $Q_{1/N}(x^2+y^2+z^2-1)$,
we conclude 
\begin{align}
\bigg{(}Q_{1/N}(q(x,y,z)) Q_{1/N}(x^2+y^2+z^2-1) \bigg{)}|_{\text{Sym}^N(\mathbb{C}^2)}=0. \label{vanishingmatrixelements}
\end{align}
Therefore, for any symmetric unit vector $\phi\in\text{Sym}^N(\mathbb{C}^2)$ we compute
\begin{align*}
&||Q_{1/N}(q(x,y,z) (x^2+y^2+z^2-1))\phi||_N=&\\&
\left|\left|\bigg{(}Q_{1/N}(q(x,y,z) (x^2+y^2+z^2-1))-Q_{1/N}(q(x,y,z))Q_{1/N}(x^2+y^2+z^2-1)\bigg{)}\phi\right|\right|_N&\\&
\leq ||Q_{1/N}(q(x,y,z) (x^2+y^2+z^2-1))-Q_{1/N}(q(x,y,z))Q_{1/N}(x^2+y^2+z^2-1)||_N\:.
\end{align*}
As a consequence of (\ref{Klaas}), for every $\epsilon>0$ there is $N_\epsilon$ such that 
$$||Q_{1/N}(q(x,y,z) (x^2+y^2+z^2-1))\phi||_N < \epsilon\quad \mbox{if $N> N_\epsilon$}$$
the crucial observation is that due to  (\ref{Klaas}) the number $N_\epsilon$ does not depend on the unit vector $\phi\in\text{Sym}^N(\mathbb{C}^2)$. Therefore the above bound is uniform, and
$$||Q_{1/N}(q(x,y,z) (x^2+y^2+z^2-1))|_{\text{Sym}^N(\mathbb{C}^2)}||_N $$
$$ = \sup_{||\phi||=1 \:, \phi\in\text{Sym}^N(\mathbb{C}^2)}||Q_{1/N}(q(x,y,z) (x^2+y^2+z^2-1))\phi||_N\leq  \epsilon\quad \mbox{if $N> N_\epsilon$},$$
which means
\begin{align*}
\lim_{N\to\infty}||Q_{1/N}(q(x,y,z) (x^2+y^2+z^2-1))|_{\text{Sym}^N(\mathbb{C}^2)}||_N=0.
\end{align*}
This closes the proof of the proposition.
\end{proof}

\section{Application to the quantum Curie-Weiss model}
We apply the previous theorem to the (quantum) Curie-Weiss model\footnote{This model exists in both a classical and a quantum version and is a mean-field approximation to the Ising model. See e.g. \cite{FV2017} for a mathematically rigorous treatment of the classical version, and \cite{CCIL,IL} for the quantum version. For our approach the papers \cite{Bona,DW,RW}  
played an important role. See also \cite{ABN} for a very detailed discussion of the quantum Curie--Weiss model.}, which
is an exemplary quantum mean-field spin model. We recall that the {\bf quantum Curie Weiss} defined on a lattice with $N$ sites is
  \begin{align}
h^{CW}_{1/N}: &  \underbrace{\mathbb{C}^2 \otimes \cdots  \otimes\mathbb{C}^2}_{N \: times}  \to 
\underbrace{\mathbb{C}^2 \otimes \cdots  \otimes\mathbb{C}^2}_{N \: times}; \\
h^{CW}_{1/N} &=\frac{1}{N} \left(-\frac{J}{2N} \sum_{i,j=1}^N \sigma_3(i)\sigma_3(j) -B \sum_{j=1}^N \sigma_1(j)\right).\label{CWham}
\end{align}
Here $\sigma_k(j)$ stands for $I_2 \otimes \cdots \otimes \sigma_k\otimes \cdots \otimes I_2$, where $\sigma_k$ occupies  the $j$-th slot, and 
 $J,B \in \mathbb{R}$ are given constants defining the strength of the spin-spin coupling and the (transverse) external magnetic field, respectively. Note that 
 \begin{equation}
 h^{CW}_{1/N} \in \mathrm{Sym}(M_2(\mathbb{C})^{\otimes N}),
\end{equation}
where $\mathrm{Sym}(M_2(\mathbb{C})^{\otimes N})$ is the range of the symmetrizer.
Our interest will lie in the limit $N\to\infty$. As such, we rewrite $h_{1/N}^{CW}$ as
\begin{align}
h^{CW}_{1/N}&= -\frac{J}{2N(N-1)} \sum^N_{i \neq j, \:i,j=1} \sigma_3(i)\sigma_3(j) - \frac{B}{N} \sum_{j=1}^N \sigma_1(j) + O(1/N).\nonumber \\
&= Q_{1/N}(h^{CW}_0)  + O(1/N),\label{QNh}
\end{align}
where $O(1/N)$ is meant in norm (i.e. the operator norm on each space $M_2(\mathbb{C}^2)^{\otimes N}$), and
the {\bf  classical Curie--Weiss Hamiltonian} is 
\begin{align}
h^{CW}_0: B^3 &\mapsto\mathbb{R};\\
h^{CW}_0(x,y,z)  &= -\left(\frac{J}{2}z^2 + Bx\right), \quad \mathbf{x}=(x,y,z) \in B^3, \label{hcwc}
\end{align}
where $B^3= \{ \mathbf{x} \in \mathbb{R}^3 \:|\: \|\mathbf{x}\|\leq 1\}$ is the closed unit ball in $\mathbb{R}^3$.\\
Using these observations we now show that the quantum Curie-Weiss Hamiltonian restricted to the symmetric space is asymptotically norm-equivalent also to the other quantization map $Q'_{1/N}$ applied to $h_0^{CW}|_{S^2}$.

\begin{theorem} One has
\begin{align}
\left|\left|h_{1/N}^{CW}|_{\text{Sym}^N(\mathbb{C}^2)}  - Q'_{1/N}(h_0^{CW}|_{S^2})\right|\right|_N \to 0 
\quad \text{for N $\to \infty$}.
\end{align}
\end{theorem}
$\null$\\

\begin{proof} Using \eqref{QNh} and Theorem \eqref{MAIN},
\begin{align}
&\left|\left|h_{1/N}^{CW}|_{\text{Sym}^N(\mathbb{C}^2)}  - Q'_{1/N}(h_0^{CW}|_{S^2})\right|\right|_N\nonumber\leq &\\&  \left|\left|Q_{1/N}(h_0^{CW})|_{\text{Sym}^N(\mathbb{C}^2)}+O(\frac{1}{N})|_{\text{Sym}^N(\mathbb{C}^2)}- Q'_{1/N}(h_0^{CW}|_{S^2})\right|\right|_N\nonumber\leq &\\&  \left|\left|Q_{1/N}(h_0^{CW})|_{\text{Sym}^N(\mathbb{C}^2)} - Q'_{1/N}(h_0^{CW}|_{S^2})\right|\right|_N\to 0 \  \ (\text{as \  $N\to\infty$}).
\end{align}
\end{proof}
This in particular establishes a link between the (compressed) quantum Curie-Weiss spin Hamiltonian and its classical counterpart on the sphere.
\section*{Acknowledgments} The authors are grateful to Riccardo Ghiloni for useful technical discussions and to Klaas Landsman for its comments and suggestions. Christiaan van de Ven is Marie Sk\l odowska-Curie fellow of the Istituto Nazionale di Alta Matematica and is funded by the INdAM Doctoral Programme in Mathematics and/or Applications co-funded by Marie 
Sk\l odowska-Curie Actions, INdAM-DP-COFUND-2015, grant number 713485. 
\appendix 
\section{Continuous bundle of $C^*$-algebras}
For any unital $C^*$-algebra $B$ the following fibers may be turned into a continuous bundle of $C^*$-algebras over the base space $I=\{0\}\cup 1/\mathbb{N}\subset[0,1]$ (with relative topology, so that $(1/N)\to 0$ as $N\to\infty$):
 \begin{align}
 A_0&=C(S(B)); \label{A0}\\
 A_{1/N}&= {B}^{\otimes N}.\label{AN}
   \end{align}
Here $S(B)$ is the (algebraic) state space of $B$ equipped with the weak$\mbox{}^*$-topology (in which it is a compact convex set,
e.g.\ the three-ball $S(M_2(\bC))\cong B^3\subset\bR^3$),
 and ${B}^{\otimes N}$ is the $N$th tensor power of $B$ also denoted by $B^N$ in what follows).\footnote{Although this is irrelevant for our main application $B=M_k(\bC)$, for general $C^*$-algebras $B$ one should equip $B^N$ with the {\em minimal} $C^*$-norm $\|\:\:\:\|_N$ \cite{Takesaki,Lan17}.} As in the case of vector bundles, the continuity structure of a bundle of $C^*$-algebras may be defined (indirectly) by specifying what the continuous cross-sections are.
To do so for \eqref{A0} - \eqref{AN}, we need the {\em symmetrization operator} $S_N : {B}^N \to {B}^N$, defined as the unique linear continuous extension of the following map on elementary tensors:
\begin{align}S_N (a_1 \otimes \cdots \otimes a_N) := \frac{1}{N!} \sum_{\sigma \in {\cal P}(N)} a_{\sigma(1)} \otimes \cdots \otimes a_{\sigma(N)}. \label{defSN}
\end{align}
Furthermore, for $N\geq M$ we need to generalize the definition of $S_N$ to give a bounded  operator  $S_{M,N}: {B}^M 	\to {B}^N$, defined by linear and continuous extension of  
\begin{align} S_{M,N}(b) := S_N(b \otimes \underbrace{I \otimes \cdots \otimes I }_{N-M \mbox{\scriptsize times}}),\quad b \in B^{\otimes M}. \label{defSMN}
\end{align}
We write cross-sections $a$ of \eqref{A0} - \eqref{AN} as sequences 
$(a_0,a_{1/N})_{N\in\bN}$, where $a(0)=a_0$ etc. Following \cite{RW}, the part of the cross-section
 $(a_{1/N})_{N\in\bN}$ away from zero (i.e.\ with $a_0$ omitted) is called {\bf symmetric}
if there exist $M \in \mathbb{N}$ and $a_{1/M} \in {B}^{\otimes M}$ such that 
\begin{align}\label{oneone}
    a_{1/N} = S_{M,N}(a_{1/M})\:\mbox{for all }  N\geq M,
\end{align}
and {\bf quasi-symmetric} if  $a_{1/N} = S_{N}(a_{1/N})$ if $N\in \mathbb{N}$,
and for every $\epsilon > 0$, there is a symmetric sequence $(b_{1/N})_{N\in\mathbb{N}}$
as well as  $M \in\mathbb{N}$ (both depending on $\epsilon$) such that 
\begin{align} \| a_{1/N}-b_{1/N}\| < \epsilon\: \mbox{ for all } N > M.
\end{align} 
The continuous cross-sections of the bundle  \eqref{A0} - \eqref{AN}, then, are the  sequences 
$(a_0,a_{1/N})_{N\in\bN}$ for which the part  $(a_{1/N})_{N\in\bN}$ away from zero is quasi-symmetric and \begin{align}
a_0(\omega)=\lim_{N\to\infty}\omega^N(a_{1/N})\label{8.46},
\end{align}
where $\omega\in S({B})$, and $\omega^N=  \underbrace{\omega\otimes \cdots \otimes \omega}_{N\: \mbox{\scriptsize times}} \in S({B}^{\otimes N})$,
is  the unique (norm) continuous linear extension of the following map that is defined on elementary tensors:
\begin{align}
\omega^N(b_1\otimes\cdot\cdot\cdot\otimes b_N)=\omega(b_1)\cdot\cdot\cdot\omega(b_N).\label{linearstateextension}
\end{align}
The limit in \eqref{8.46} exists provided  $(a_{1/N})_{N\in\bN}$ is quasi-symmetric (as we assume), and by \cite[Theorem 8.4]{Lan17}, this choice of continuous cross-sections uniquely defines (or identifies) 
 a continuous bundle of $C^*$-algebras over  $I$ in \eqref{defI} with fibers \eqref{A0} - \eqref{AN}.

\section{Coherent spin states}\label{appB}
 If $|\!\uparrow\rangle, |\!\downarrow\rangle$ are the eigenvectors of $\sigma_3$ in $\mathbb{C}^2$, so that
$\sigma_3|\!\uparrow\rangle=|\!\uparrow\rangle$ and $\sigma_3|\!\downarrow\rangle=- |\!\downarrow\rangle$, and where  $\Omega \in {S}^2$, with  polar angles  
$\theta_\Omega \in (0,\pi)$, $\phi_\Omega \in (-\pi, \pi)$, we then define the unit vector
\begin{align}\label{om1}
|\Omega\rangle_1  = \cos \frac{\theta_\Omega}{2} |\!\uparrow\rangle + e^{i\phi_\Omega}\sin   \frac{\theta_\Omega}{2} |\!\downarrow\rangle.
\end{align}
 If $N \in \mathbb{N}$, the associated {\bf $N$-coherent spin state} $|\Omega\rangle_N\in \mathrm{Sym}^N(\mathbb{C}^2)$, equipped with the usual scalar product $\langle \cdot ,\cdot \rangle_N$ inherited from    
$(\mathbb{C}^2)^N$,  is defined as follows \cite{Pe72}:
\begin{align}\label{om2}
|\Omega\rangle_N =  \underbrace{|\Omega\rangle_1 \otimes \cdots \otimes |\Omega \rangle_1}_{N \: times}.
\end{align}
An important property relevant for our computations was established in \cite{LMV} 
\begin{align}
f(\Omega')=\lim_{N\to\infty}\frac{N+1}{4\pi}\int_{S^2}d\Omega f(\Omega')|\langle \Omega,\Omega' \rangle_N|^2, \ \ (f\in C(S^2))\:.\label{Berezinprop}
\end{align}

\end{document}